\documentclass{acm_proc_article-sp}
\usepackage{algorithm,balance,color,graphicx,pgfplots,subfig,tikz}
\usepackage[noend]{algorithmic}
\usepackage[english]{babel}

\usetikzlibrary{shapes.misc}
\tikzset{
    graph/.style={
        every edge/.style={
            very thick,
            draw=black!75
        },
        every node/.style={
            rounded rectangle,
            minimum size=6mm,
            very thick,
            draw=black!50,
            top color=white,
            bottom color=black!20,
            font=\ttfamily
        }
    }
}

\newcommand{\SPACE}{\vspace*{0.65em}}

\newcommand{\children}{\mathrm{children}}
\newcommand{\parents}{\mathrm{parents}}

\newcommand{\bisim}{\approx}

\newcommand{\rank}{\mathit{rank}}

\newcommand{\Scan}{\textsc{Scan}}
\newcommand{\Sort}{\textsc{Sort}}

\newcommand{\ignore}[1]{}

\newcommand{\bsm}{bisimilarity}
\newcommand{\id}{\mathit{id}}

\newtheorem{theorem}{Theorem}
\newtheorem{proposition}{Proposition}

\newdef{definition}{Definition}

\begin{document}

\title{I/O efficient bisimulation partitioning on very large directed acyclic graphs}

%
%
%
%
%

\numberofauthors{3} 
%
\author{
%
%
\alignauthor
Jelle Hellings\\
       \affaddr{Hasselt University}\\
       \affaddr{Belgium}\\
       \email{exbisim@jhellings.nl}
\alignauthor
George H.L.\ Fletcher\\
       \affaddr{Eindhoven University of Technology}\\
       \affaddr{The Netherlands}\\
       \email{g.h.l.fletcher@tue.nl}
\alignauthor
Herman Haverkort\\
       \affaddr{Eindhoven University of Technology}\\
       \affaddr{The Netherlands}\\
       \email{cs.herman@haverkort.net}
}

\date{\today}

\maketitle
\pagestyle{plain}
\begin{abstract}

In this paper we introduce the first efficient external-memory algorithm to
compute the bisimilarity equivalence classes of a directed acyclic graph (DAG).
DAGs are commonly used to model data in a wide variety of practical
applications, ranging from XML documents and data provenance models, to web
taxonomies and scientific workflows.   In the study of efficient reasoning over
massive graphs, the notion of node bisimilarity plays a central role.  For
example, grouping together bisimilar nodes in an XML data set is the first step
in many sophisticated approaches to building indexing data structures for
efficient XPath query evaluation.    To date, however, only internal-memory
bisimulation algorithms have been investigated.   As the size of real-world DAG
data sets often exceeds available main memory, storage in external memory
becomes necessary.  Hence, there is a practical need for an efficient approach
to computing bisimulation in external memory.

Our general algorithm has a worst-case  IO-complexity of $O(\Sort(|N| + |E|))$, where
$|N|$ and $|E|$ are the numbers of nodes and edges, resp., in the data graph and
$\Sort(n)$ is the number of accesses to external memory needed to sort an input
of size $n$.  We also study specializations of this algorithm to common
variations of bisimulation for tree-structured XML data sets.  We empirically
verify efficient performance of the algorithms on graphs and XML documents
having billions of nodes and edges, and find that the algorithms can process
such graphs efficiently even when very limited internal memory is available. The
proposed algorithms are simple enough for practical implementation and use, and
open the door for further study of external-memory bisimulation algorithms.   To
this end, the full open-source C++ implementation has been made freely
available.

\end{abstract}

\section{Introduction}\label{sec:intro}

Data modeled as directed acyclic graphs (DAGs) arise in a
diversity of practical applications such as biological and biomedical ontologies
\cite{obo}, web folksonomies \cite{palla},  scientific workflows \cite{yu},
semantic web schemas \cite{ChristophidesKPST04}, business process modeling
\cite{Ben-AriMV09,deutch}, data provenance modeling
\cite{Moreau10,Muniswamy-Reddy:2010}, and the widely adopted XML standard
\cite{fbbook}.  It is anticipated that the variety, uses, and quantity of
DAG-structured data sets will only continue to grow in the future.

In each of these application areas, efficient searching and querying on the data
is a basic challenge.  In reasoning over massive data sets, typically index data
structures are computed and maintained to accelerate processing. These indexes
are essentially a reduction or summary of the underlying data.  Efficiency is
achieved by performing reasoning over this reduction to the extent possible,
rather than directly over the original data.

Many approaches to indexing have been investigated in preceding decades.
Reductions of data sets typically group together data elements based on their
shared values or substructures in the data.  In graphs, the notion of {\em
bisimulation equivalence} of nodes has proven to be an effective means for
indexing (e.g.,
\cite{fbbook,FletcherGWGBP09,linear,fbkaushik,akindex,indexbb,fbdisk}).
Bisimulation, which is a fundamental notion arising in a surprising range of
contexts \cite{bisim}, is based on the structural similarity of subgraphs.
Intuitively, two nodes are bisimilar to each other if they cannot be
distinguished from each other by the sequences of node labels that may appear on
the paths that start from these nodes, as well as from each of the nodes on those paths.
Grouping bisimilar nodes is known as {\em bisimulation partitioning}.  Blocks of
bisimilar nodes are then used as the basis for constructing indexing data
structures supporting efficient search and querying over the data.

Efficient {\em internal}-memory solutions for computing bisimulation partitions
have been investigated (e.g., \cite{fastbisim,linear,pt}).  To scale to
real-world data sets such as those discussed above, it becomes necessary to
consider DAGs resident in {\em external} memory.  In considering algorithms for
such data, the primary concern is to minimize disk IO operations due to the high
cost involved, relative to main-memory operations, in performing reads and writes
to disk.

Due to the random access nature of internal-memory algorithms, the design of
external-memory algorithms which minimize disk IO typically requires a
significant departure from approaches taken for internal memory solutions \cite{meyer}.
In particular, state-of-the-art internal-memory bisimulation algorithms can not be directly adapted to IO-efficient external-memory algorithms due to their inherent random access behaviour.
While a study has been made on storing and querying bisimulation partitions
on disk \cite{fbdisk}, there has been to our knowledge no approach developed to
date for efficiently computing bisimulation partitioning in external memory.

Motivated by these observations, in this paper we give the first IO-efficient
external-memory bisimulation algorithm for DAGs.
Our algorithm has a worst-case IO-complexity of $O(\Sort(|N| + |E|))$,
where $|N|$ and $|E|$ are the number of nodes and edges, resp., in the data
graph and $\Sort(n)$ is the number of accesses to external memory needed
to sort an input of size $n$.
Efficiency is achieved by
intelligent organization of the graph on disk, and by sophisticated processing
of the graph using global and local reorganization and careful staging and use
of local bisimulation information.  We establish the theoretical efficiency of
the algorithm, and demonstrate its practicality via a thorough empirical
evaluation on data sets having billions of nodes and edges.

Our algorithm is simple enough for practical implementation and use, and to
serve as the basis for further study and design of external-memory bisimulation
algorithms.  For example, we also develop in this paper specializations of our
algorithm for computing common variations of bisimulation for tree-structured
graphs in the form of XML documents.  Furthermore, the complete
implementation is open-source and available for download.

We proceed in the paper as follows.   In the next section, we present basic
definitions concerning our data model, bisimulation equivalence, and the
standard external-memory computational model.   In Section
\ref{sec:partitioning}, we then present and theoretically analyze our
external-memory bisimulation algorithm.  In Section \ref{sec:xml}, we show how
to specialize our general algorithm for various bisimulation notions proposed
for XML data.  In Section \ref{sec:exp}, we then present a
thorough empirical analysis of our approach, and conclude in Section
\ref{sec:conclude} with a discussion of future directions for research.

\section{Preliminaries}
\label{sec:prelim}

\subsection{Graphs and bisimilarity}



In the context of this paper, a graph $G$ is a triple $G = \langle N, E, l\rangle$, where $N$ is a finite set of nodes, $E \subseteq N \times N$ is a directed edge relation, and $l$ is a function with domain $N$ that assigns a label $l(n)$ to every node $n \in N$.
With a slight abuse of terminology, we call $n$ a child of $m$, and $m$ a parent of $n$, if and only if $G$ contains an edge $(m,n)$.
Let $\children(m)$ be the set of all children of $m$, and let $\parents(n)$ be the set of all parents of $n$. Note that in our work we only
consider acyclic graphs. Furthermore, we assume that the node set $N$ is ordered in reverse topological order, that is, children always precede their parents in the order.
Assuming a topological ordering is standard in the design of external 
memory DAG algorithms \cite{meyer}.   Indeed, real world data is often 
already ordered (e.g., XML documents), and, furthermore, practical 
approaches to topological sorting of massive data sets are available 
\cite{toposort}.

\begin{definition}\label{def:bisim-nodes}
Let $G_1 =\langle N_1, E_1, l_1\rangle$ and $G_2 = \langle N_2, E_2, l_2\rangle$ be two, possibly the same, graphs.
Nodes $n_1 \in N_1$ and $n_2 \in N_2$ are \emph{bisimilar} to each other, denoted $n_1 \bisim n_2$, 
if and only if:
\begin{enumerate}
    \item the nodes have the same label: $l_1(n_1) = l_2(n_2)$;
    \item for every node $n'_1 \in \children(n_1)$ there is a node $n'_2 \in \children(n_2)$ such that $n'_1 \bisim n'_2$, and:
    \item For every node $n'_2 \in \children(n_2)$ there is a node $n'_1 \in \children(n_1)$ such that $n'_1 \bisim n'_2$.
\end{enumerate}
\end{definition}

We can extend this notion to complete graphs as follows:


\begin{definition}\label{def:bisim-graphs}
Let $G_1 =\langle N_1, E_1, l_1\rangle$ and $G_2=\langle N_2, E_2, l_2\rangle$ be graphs. Graph $G_1$ and $G_2$ are bisimilar to each other, denoted as $G_1 \bisim G_2$, if and only if:
\begin{enumerate}
    \item For every node $n_1 \in N_1$ there is a node $n_2 \in N_2$ such that $n_1 \bisim n_2$, and
    \item For every node $n_2 \in N_2$ there is a node $n_1 \in N_1$ such that $n_1 \bisim n_2$.
\end{enumerate}
\end{definition}

Figure \ref{fig:exbisimgraph} shows two graphs that are bisimilar to each other. The figure also shows with
dotted lines how the nodes of one graph are bisimilar to nodes of the other graph. Note that in this figure all nodes with label {\tt a} are bisimilar to each other, all nodes with label {\tt b} are bisimilar to each other, and all nodes with label {\tt c} are bisimilar to each other. 
Note, however, that this does not hold for nodes with label {\tt d}.

\begin{figure}
    \centering
        \begin{tikzpicture}[graph]
            \node (n1) at (1.5, 2) {a};
            \node (n2) at (2, 1) {b};
            \node (n3) at (1, 1) {b};
            \node (n4) at (1, 0) {c};
            \node (n5) at (2, 0) {c};
            \node (n9) at (0, 2) {d};
            \node (n10) at (0, 1) {d};

            \path[->] (n1) edge (n2) edge (n3)
                      (n2) edge (n4) edge (n5)
                      (n3) edge (n4)
                      (n9) edge (n3)
                      (n10) edge (n4);

            \node (n6) at (4, 2) {a};
            \node (n7) at (4, 1) {b};
            \node (n8) at (4, 0) {c};
            \node (n11) at (5, 2) {d};
            \node (n12) at (5, 1) {d};

            \path[->] (n6) edge (n7)
                      (n7) edge (n8)
                      (n11) edge (n7)
                      (n12) edge (n8);

            \path[dashed] (n6) edge[thick,draw=red!50] (n1)
                          (n7) edge[thick,draw=red!50,bend right=25] (n3)
                               edge[thick,draw=red!50,bend left=15]  (n2)
                          (n8) edge[thick,draw=red!50,bend right=25] (n4)
                               edge[thick,draw=red!50,bend left=15]  (n5)
                          (n9) edge[thick,draw=red!50,bend left=25] (n11)
                          (n10) edge[thick,draw=red!50,bend left=25] (n12);
        \end{tikzpicture}
    \caption{Two bisimilar directed acyclic graphs.}
    \label{fig:exbisimgraph}
\end{figure}
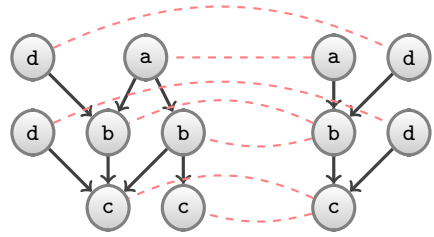

For each graph $G$ there is a unique smallest graph (having the fewest nodes) that is bisimilar to $G$; we call this smallest graph the \emph{maximum bisimulation graph} of $G$ and denote it by $G_\downarrow = \langle N_\downarrow, E_\downarrow, l_\downarrow\rangle$. In Figure \ref{fig:exbisimgraph}, the graph on the right is the maximum bisimulation graph of itself. It is also the maximum bisimulation graph of the graph on the left. 

In the context of this paper, the \emph{bisimilarity index} of a graph $G$ is a data structure that stores the maximum bisimilarity graph $G_\downarrow$ of $G$, and stores, for each node $n_\downarrow \in N_\downarrow$, the set of nodes $\{n \in N : n_\downarrow \bisim n\}$,
i.e., the bisimulation equivalence class of $n_\downarrow$.


A \emph{partition} ${\cal P}$ of a graph $G = \langle N, E, l\rangle$ is a subdivision of its nodes $N$ into a set of \emph{blocks} ${\cal P} = \{N_1, N_2, ...\}$ such that each block $N_i \in {\cal P}$ is a subset of $N$, the blocks are mutually disjoint, and their union is $N$. A \emph{bisimulation partition} of a graph $G$ is a partition ${\cal P}$ of $G$ such that the blocks of ${\cal P}$ are exactly the bisimilarity equivalence classes of $G$.




A partition ${\cal P}_1$ is a \emph{refinement} of a partition ${\cal P}_2$ if and only if for every $P_1 \in {\cal P}_1$ there is exactly one $P_2 \in {\cal P}_2$ such that $P_1 \subseteq P_2$.



The rank $\rank(n)$ of a node $n$ is defined as the maximum number of edges on any path that starts at $n$. It is easily proved by induction that $m \bisim n$ implies $\rank(m) = \rank(n)$, and thus the bisimulation partition is a refinement of the partition by rank.

\subsection{Analysis of external-memory algorithms}
In this paper, we investigate algorithms operating on data that does not fit in main memory. Therefore we need to use external memory, such as disks. In general external memory is slow. In particular, there is a high latency: it takes a lot of time to start reading or writing a random data item in external memory, but after that a large block that is consecutive in external memory can be read or written relatively fast. Thus the performance of algorithms using external memory is often dominated by the external-memory access patterns, where algorithms that read from and write to disk sparingly and in large blocks are at an advantage over algorithms that access the disk often for small amounts of data.

We shall use the following standard computer model to analyze the efficiency of our algorithms~\cite{Aggarwal}. Our computer has a fast memory with a limited size of $M$ units of data, and a slow, external memory of practically unlimited size. The computer has a fast processing unit that can operate on data in fast memory, but not on data in external memory. Therefore, during operation of any algorithm on this computer, data needs to be transferred between the two memories. This is done by moving data in blocks of size $B$; such a transfer is called an \emph{IO}. The block size $B$ is assumed to be large enough that the latency is dominated by the actual transfer times, and thus, the time spent on external-memory access is roughly proportional to the number of IOs.
%
%
%

The complexity of an external-memory algorithm can now be expressed as the (asymptotic) number of IOs performed by an algorithm, as a function of the input size and, possibly, other parameters. Clearly, reading or writing $n$ units of data that are (to be) stored consecutively in external memory takes $\Theta(\Scan(n)) = \Theta(\frac{n}{B})$ IOs. Sorting $n$ units of data that are consecutive in external memory takes $\Theta(\Sort(n)) = \Theta(\frac{n}{B}\log_{M/B}(\frac{n}{B}))$ IOs~\cite{Aggarwal}.

\section{Bisimulation Partitioning}
\label{sec:partitioning}

State of the art internal memory bisimulation algorithms are based on a process
of refinement introduced in the work of Paige and Tarjan (e.g.,
\cite{fastbisim,linear,pt}).  An initial partition of the nodes is picked (for
example: a partition based on label equivalence). A step-by-step refinement of
this initial partition is calculated by picking a single block $S$ of nodes from
the partition and stabilize all other blocks $B$ with respect to this group (by
splitting $B$ into a block of nodes that have children in $S$ and into a block
of nodes that do not have children in $S$). 

These refinement steps require unstructured random access to nodes and their children.
In an external memory setting, these accesses translate to high IO costs.
Therefore, it is not clear that state-of-the-art internal memory bisimulation
algorithms can be effectively adapted to an external-memory setting.  Hence, we
have chosen to investigate an alternative approach, inspired by the recent use
of node rank to accelerate internal memory refinement computations
\cite{fastbisim,linear}.

\subsection{Outline of our approach}

Suppose each bisimilarity equivalence class is identified by a unique number,
the \emph{bisimilarity identifier}. Let now the \emph{\bsm\ family} of a node be
the set of \bsm\ identifiers of its children, and let the \emph{\bsm\ decision
value} of a node be the combination of its rank, its label, and its \bsm\
family. Then, by Definition~\ref{def:bisim-nodes}, all nodes in the same
bisimilarity equivalence class have the same \bsm\ decision value, and each
bisimilarity equivalence class is uniquely identified by the \bsm\ decision
value of its nodes. 

The main idea behind our algorithmic approach is now to match \bsm\ identifiers
to nodes and their \bsm\ decision values, by processing the nodes in order of
increasing rank. Thus, when processing the nodes of any rank $r$, the \bsm\
identifiers of the children of these nodes are already known and can be used to
determine the \bsm\ decision values of the nodes of rank $r$, which can then be
sorted in order to assign a unique identifier to each different \bsm\ decision
value.

To implement this approach, we use an algorithm in two phases. In the first phase, we compute the ranks of all nodes and we sort the nodes by rank and label; in the second phase, we obtain the \bsm\ family for each node and sort nodes of equal rank and label by their families. Below we will explain how these phases can be implemented to run in $O(\Sort(|N|+|E|))$ IOs. After that, we will present an enhanced version of the algorithm where, in the first phase, nodes are sorted not only by rank and label, but also by a recursively defined hash value. This enhancement leads to a small increase in cost of the first phase, but may result in a substantial reduction in the size of the sets of nodes that need to be sorted in the second phase. Thus the enhanced algorithm still takes $O(\Sort(|N|+|E|))$ IOs, but with better constant factors in practice for certain types of inputs.

\subsection{Time-forward processing}
Our bisimulation partitioning algorithm has two phases in which information about nodes must be computed from information about their children: in the first phase, we need to compute a node's rank (which is one plus the maximum rank of its children); in the second phase, we need to assign \bsm\ identifiers to nodes based on the \bsm\ identifiers of their children. This would be relatively easy if we could access the children of any node $n$ when we process $n$, but in an external-memory setting, this could cause many IOs.

\begin{algorithm}[t]
\caption{Phase 1: sort by rank and label}\label{alg:sortbyrank}
\begin{algorithmic}[1]
\REQUIRE file of nodes $N$ as records $(\mathit{id}, \mathit{label})$, sorted by $\mathit{id}$;\\
         file of edges $E$ as records $(\mathit{parent}, \mathit{child})$, sorted by $\mathit{child}$;\\
         $(\id(n),\id(m)) \in E$ implies $\id(m) < \id(n)$.
\ENSURE  file of nodes $N'$ as records $(\mathit{id}, \mathit{origId}, \mathit{rank}, \mathit{label})$,
         sorted by $\mathit{id}$ and, simultaneously, by $(\mathit{rank}, \mathit{label})$;\\
         file of edges $E'$ as records $(\mathit{parent}, \mathit{child})$, sorted by $\mathit{child}$;\\
         $\rank(n) > \rank(m)$ implies $\id(m) < \id(n)$.
\SPACE
\STATE create empty file $\mathit{Ranks}$ of records $(\mathit{id}, \mathit{rank}, \mathit{label})$
\STATE create empty priority queue $Q$ of records\\$(\mathit{id}, \mathit{childsrank})$, ordered by $\mathit{id}$
\SPACE
\FORALL{$(n, \mathit{label}) \in N$, in order}
    \STATE $\mathit{rank} \gets 0$
    \WHILE{record at head of $Q$ has $\mathit{id} = n$}
        \STATE extract $(n, \mathit{childsrank})$ from $Q$
        \STATE $\mathit{rank} \gets \max(\mathit{rank}, \mathit{childsrank}+1)$
    \ENDWHILE
    \STATE append $(n, \mathit{rank}, \mathit{label})$ to $\mathit{Ranks}$
    \WHILE{next edge from $E$ has $\mathit{child} = n$}
        \STATE read $(\mathit{parent}, n)$ from $E$
        \STATE insert $(\mathit{parent}, \mathit{rank})$ in $Q$
    \ENDWHILE
\ENDFOR
\SPACE
\STATE sort $\mathit{Ranks}$ lexicographically by $\mathit{rank}, \mathit{label}$
\STATE copy $\mathit{Ranks}$ to $N'$ while assigning new node identifiers
\STATE copy $E$ to $E'$ while updating node identifiers in $E'$
\STATE sort $E'$ by $\mathit{child}$
\SPACE
\RETURN $(N', E')$
\end{algorithmic}
\end{algorithm}

We can remove explicit access to the children of a node by introducing an
IO-efficient supporting data structure that can be used to send information from
children to parents. This technique is called {\em time-forward processing}
\cite{chiang,meyer}. Time-forward processing can be used when nodes have unique
ordered node identifiers such that children have smaller identifiers than their
parents, and the nodes are stored in order of their identifiers, each node being
stored with its own identifier and those of its parents. The supporting data
structure should support two operations: (i) inserting a message addressed to a
given node, identified by its node identifier, and (ii) inspecting and removing all
messages addressed to the smallest node identifier that is currently present in
the data structure.

An algorithm that computes a value for each node depending
on the values of its children can now be implemented as follows. We compute
values for all nodes in order of their identifier, and whenever we compute a
node's value, we insert messages with that value in the supporting data
structure, addressing these messages to each of the node's parents. Thus, before
we process each node $n$, we can obtain the values computed for its children by
extracting all messages addressed to $n$ from the data structure. Each node
removes all messages addressed to it from the data structure, nodes with lower
identifiers are processed before nodes with higher identifiers, and no messages
are ever addressed to nodes that have already been processed; thus, when we want
to extract the messages addressed to $n$, these messages will be the messages
with the smallest node identifier currently in the data structure and they can
be extracted by an operation of type (ii).

The supporting data structure can be implemented as a priority queue. There are external-memory priority queues that, amortized over their life-time, perform $k$ operations of type (i) and (ii) in $\Theta(\Sort(k))$ IOs~\cite{Arge}.

\subsection{The bisimulation partitioning algorithm}\label{subsec:basicalgo}

\begin{algorithm}[t]
\caption{Details of line 13 and 14 of Algorithm~\ref{alg:sortbyrank}}\label{alg:updateidentifiers}
\begin{algorithmic}[1]
\STATE $\mathit{newId} \gets 0$
\STATE create empty file $R$ of records $(\mathit{origId}, \mathit{newId})$
\STATE create empty file $N'$ of records \\
       $(\mathit{newId}, \mathit{origId}, \mathit{rank}, \mathit{label})$
\FORALL{$(\mathit{origId}, \mathit{rank}, \mathit{label}) \in \mathit{Ranks}$, in order}
    \STATE $\mathit{newId} \gets \mathit{newId} + 1$
    \STATE append $(\mathit{origId}, \mathit{newId})$ to $R$
    \STATE append $(\mathit{newId}, \mathit{origId}, \mathit{rank}, \mathit{label})$ to $N'$
\ENDFOR
\SPACE
\STATE sort $R$ by $\mathit{origId}$
\STATE create empty file $E'$ of records $(\mathit{parent},\mathit{child})$
\STATE move read pointer of $E$ to beginning
\FORALL{$(\mathit{origId}, \mathit{newId}) \in R$, in order}
    \WHILE{next edge from $E$ has $\mathit{child} = \mathit{origId}$}
        \STATE read $(\mathit{parent}, \mathit{origId})$ from $E$
        \STATE append $(\mathit{parent}, \mathit{newId})$ to $E'$
    \ENDWHILE
\ENDFOR
\STATE sort $E'$ by $\mathit{parent}$
\STATE move pointers of $R$ and $E'$ to beginning
\FORALL{$(\mathit{origId}, \mathit{newId}) \in R$, in order}
    \WHILE{next edge from $E$ has $\mathit{parent} = \mathit{origId}$}
        \STATE read record $(\mathit{origId}, \mathit{child})$ from $E'$ and
        \STATE overwrite with record $(\mathit{newId}, \mathit{child})$
    \ENDWHILE
\ENDFOR
\end{algorithmic}
\end{algorithm}

Assume the input to our problem consists of a list of nodes $N$, storing a
unique node identifier and a label for every node, and a list of edges $E$,
specified by the node identifiers of their tails (parents) and their heads
(children). The list $N$ is sorted by node identifier, and the list $E$ is
sorted by head (child). Recall that the node identifiers are assumed to be such
that children always have smaller node identifiers than their parents.

Our basic bisimulation partitioning algorithm is now as follows. We use time-forward processing to compute the rank of each node, and make a copy of the list of nodes in which each node is annotated with its rank. Then we sort the nodes lexicographically, with their ranks as primary keys and their labels as secondary keys. We give each node a new identifier which is simply the position of the node in the resulting sorted list, and we replace the identifiers in $E$ accordingly, producing lists $N'$ and $E'$. We sort these new lists by the (new) node identifiers and by the (new) node identifiers of the heads, respectively. This completes the first phase of the algorithm. Pseudocode for this phase is given in Algorithm~\ref{alg:sortbyrank}.


\begin{algorithm*}
\caption{Phase 2: sort by \bsm\ equivalence class}\label{alg:exbisim}
\begin{algorithmic}[1]
\REQUIRE file of nodes $N'$ as records $(\mathit{id}, \mathit{origId}, \mathit{rank}, \mathit{label})$,
         sorted by $\mathit{id}$ and, simultaneously, by $(\mathit{rank}, \mathit{label})$;\\
         file of edges $E'$ as records $(\mathit{parent}, \mathit{child})$, sorted by $\mathit{child}$;\\
         $\rank(n) > \rank(m)$ implies $\id(m) < \id(n)$.
\ENSURE  file of nodes $B$ as records $(\mathit{origId}, \mathit{bisimId})$
\SPACE
\STATE create empty file $B$ of records $(\mathit{origId}, \mathit{bisimId})$
\STATE create priority queue $Q$ of records $(\mathit{id}, \mathit{childsBisimId})$, ordered by $\mathit{id}$
\STATE $\mathit{lastBisimId} \gets 0$
\STATE create empty file $\mathit{Group}$ of records $(\mathit{bisimFamily}, \mathit{origId}, \mathit{parents})$
\SPACE
\FORALL{$(n, \mathit{origId}, r, l) \in N'$, in order}
    \STATE create an empty list $\mathit{bisimFamily}$
    \WHILE{record at head of $Q$ has $\mathit{id} = n$}
        \STATE extract $(n, \mathit{childsBisimId})$ from $Q$
        \STATE append $\mathit{childsBisimId}$ to $\mathit{bisimFamily}$
    \ENDWHILE
    \STATE sort $\mathit{bisimFamily}$, removing doubles
    \STATE read all parents of $n$ from $E'$ and put them in a list $\mathit{parents}$
    \STATE append $(\mathit{bisimFamily}, \mathit{origId}, \mathit{parents})$ to $\mathit{Group}$
    \SPACE
    \IF{$N'$ has no more records with $\mathit{rank} = r$, $\mathit{label} = l$}
        \STATE sort $\mathit{Group}$ by $\mathit{bisimFamily}$,
               while marking the first occurrence of each family
        \FORALL{$(\mathit{bisimFamily}, \mathit{origId}, \mathit{parents}) \in \mathit{Group}$, in order}
            \IF{$\mathit{bisimFamily}$ is marked}
                \STATE $\mathit{lastBisimId} \gets \mathit{lastBisimId} + 1$
            \ENDIF
            \STATE append $(\mathit{origId}, \mathit{lastBisimId})$ to $B$
            \FORALL{$\mathit{parentId} \in \mathit{parents}$}
                \STATE insert $(\mathit{parentId}, \mathit{lastBisimId})$ in $Q$
            \ENDFOR
        \ENDFOR
        \STATE erase $\mathit{Group}$
    \ENDIF
\ENDFOR
\SPACE
\RETURN $B$
\end{algorithmic}
\end{algorithm*}

Some additional implementation details on the last lines of Algorithm
\ref{alg:sortbyrank} are in order. We can copy $\mathit{Ranks}$ to $N'$ while
assigning new node identifiers, going through $\mathit{Ranks}$ in order. During
this process we construct a list $R$ of (old node identifier, new node
identifier)-pairs.  To obtain a list $E'$ with updated child node identifiers,
we scan $E$ and $R$ in parallel from beginning to end, copying the entries of
$E$ to $E'$ while replacing the child node identifiers by the new identifiers as
read from $R$.  To update the parent node identifiers in $E'$ we sort $E'$ on
parent node identifier; we then scan $E'$ and $R$ in parallel from beginning to
end while replacing the parent node identifiers in $E'$ by the new identifiers
as read from $R$. Pseudocode is given in Algorithm~\ref{alg:updateidentifiers}.


The rank-label combinations of the nodes define a partitioning of the graph. In the second phase of the algorithm, we use time-forward processing to go through the blocks of this partitioning one by one. Each rank-label combination $c$ is processed as follows. Let $N_c$ be the set of nodes that have rank-label combination $c$. For each node of $N_c$, we extract the \bsm\ identifiers of its children from the priority queue (assuming that they have been placed there) and sort them, while removing doubles. Thus we get the \bsm\ families for all nodes of $N_c$. Then we sort the nodes of $N_c$ by \bsm\ family. Finally we go through the nodes of $N_c$ in order, assigning a unique \bsm\ identifier $\mathit{bisimId}(f)$ to each maximal group of nodes $N_f$ within $N_c$ that have the same \bsm\ family~$f$, and putting a message $\mathit{bisimId}(f)$ in the priority queue for all parents of the nodes of $N_f$.
Pseudocode for the second phase of the algorithm is given in Algorithm~\ref{alg:exbisim}.



\begin{theorem}\label{thm:dagfps_exp}
Given a labeled directed acyclic graph $G = \langle N, E, l\rangle$ with its nodes numbered in (reverse) topological order, we can compute the bisimilarity equivalence classes of $G$ in $O(\Sort(|N| + |E|))$ IOs.
\end{theorem}
\begin{proof}
We use Algorithm~\ref{alg:sortbyrank}, followed by Algorithm~\ref{alg:exbisim}.
As observed in Section~\ref{sec:prelim}, bisimilar nodes must have the same rank and the same label.
As a result, any nodes that are bisimilar to each other are processed in the same execution of lines 14--21 of Algorithm~\ref{alg:exbisim}. Based on the induction hypothesis that nodes of rank $r-1$ get the same \bsm\ identifier if and only if they are bisimilar to each other, it is now easy to show that in lines 14--21, nodes of rank $r$ get the same \bsm\ identifier if and only if they are bisimilar to each other.

As for the efficiency of the algorithm, the first phase scans and sorts files of at most $|N| + |E|$ records, for a total of $O(\Sort(|N| + |E|))$ IOs. One record is inserted into and extracted from the priority queue for each child-parent relation; thus the total number of IOs required by the priority queue is $O(\Sort(|E|))$.

The second phase is slightly more involved, as it sorts the lists $\mathit{bisimFamily}$ and the files $\mathit{Group}$. For each node, the list $\mathit{bisimFamily}$ contains one entry for each edge originating from that node; thus the total size of the lists $\mathit{bisimFamily}$ is $O(|E|)$ and they are sorted in $O(\Sort(|E|))$ IOs in total. For each node $n$, one record is added to the file $\mathit{Group}$: this records contains an identifier for $n$ and each of its children and parents. Thus the amount of data inserted into $\mathit{Group}$ over the course of the entire algorithm is $O(|N| + |E|)$. On line 14, the variable-size records in $\mathit{Group}$ can be sorted and marked with the string sorting algorithm by Arge et al.~\cite{AFGV} in $O(\Sort(|N| + |E|))$ IOs.
Thus, the complete algorithm takes $O(\Sort(|N| + |E|))$ IOs.
\end{proof}

\subsection{Enhanced algorithm}\label{subsec:enhancedalgo}

To reduce the amount of sorting needed in the second phase of the algorithm, we
propose the following enhanced algorithm. In the first phase, we not only
compute a rank for each node, but also a {\em hash} value, which is computed from the
node's label and from the hash values of its children. Thus, the first phase of
the algorithm is as in Algorithm~\ref{alg:sortbyhash}.

\begin{algorithm}[htb!]
\caption{Phase 1 (enhanced with hash values)}\label{alg:sortbyhash}
\begin{algorithmic}[1]
\REQUIRE file of nodes $N$ as records $(\mathit{id}, \mathit{label})$, sorted by $\mathit{id}$;\\
         file of edges $E$ as records $(\mathit{parent}, \mathit{child})$, sorted by $\mathit{child}$;\\
         $(\id(n),\id(m)) \in E$ implies $\id(m) < \id(n)$.
\ENSURE  nodes $N'$ as records $(\mathit{id}, \mathit{origId}, \mathit{rank}, \mathit{label}, \mathit{hash})$,
         sorted by $\mathit{id}$ and, simultaneously, by $(\mathit{rank}, \mathit{label}, \mathit{hash})$;\\
         file of edges $E'$ as records $(\mathit{parent}, \mathit{child})$, sorted by $\mathit{child}$;\\
         $\rank(n) > \rank(m)$ implies $\id(m) < \id(n)$.
\SPACE
\STATE create empty file $\mathit{Ranks}$ of records $(\mathit{id}, \mathit{rank}, \mathit{label}, \mathit{hash})$
\STATE create empty priority queue $Q$ of records\\ $(\mathit{id}, \mathit{childsRank}, \mathit{childsHash})$, ordered by $\mathit{id}$
\SPACE
\FORALL{$(n, \mathit{label}) \in N$, in order}
    \STATE $\mathit{rank} \gets 0$
    \STATE initialize empty list $\mathit{childrensHashes}$
    \WHILE{record at head of $Q$ has $\mathit{id} = n$}
        \STATE extract $(n, \mathit{childsRank}, \mathit{childsHash})$ from $Q$
        \STATE $\mathit{rank} \gets \max(\mathit{rank}, \mathit{childsRank}+1)$
        \STATE add $\mathit{childsHash}$ to $\mathit{childrensHashes}$
    \ENDWHILE
    \STATE sort $\mathit{childrensHashes}$, removing doubles
    \STATE $\mathit{hash} \gets$ hash value from $\mathit{label}$ and $\mathit{childrensHashes}$
    \STATE append $(n, \mathit{rank}, \mathit{label}, \mathit{hash})$ to $\mathit{Ranks}$
    \WHILE{next edge from $E$ has $\mathit{child} = n$}
        \STATE read $(\mathit{parent}, n)$ from $E$
        \STATE insert $(\mathit{parent}, \mathit{rank})$ in $Q$
    \ENDWHILE
\ENDFOR
\SPACE
\STATE sort $\mathit{Ranks}$ lexicographically by $\mathit{rank}, \mathit{label}, \mathit{hash}$
\STATE copy $\mathit{Ranks}$ to $N'$ while assigning new node identifiers
\STATE copy $E$ to $E'$ while updating node identifiers in $E'$
\STATE sort $E'$ by $\mathit{child}$
\RETURN $(N', E')$
\end{algorithmic}
\end{algorithm}

The second phase of the algorithm is exactly as before, except that $\mathit{rank}, \mathit{label}$ is replaced by $\mathit{rank}, \mathit{label}, \mathit{hash}$; in particular, lines 14--21 are executed each time $N'$ has no more records with the same rank, label, \emph{and hash value} as the records seen so far.

By induction on increasing rank one can prove that bisimilar nodes get the same hash value, and therefore, any pair of nodes that are bisimilar to each other will still be processed in the same execution of lines 14--21 in Algorithm~\ref{alg:exbisim}. Thus, the algorithm is still correct. Note that if the hash value from a label and any given set of $k$ hash values can be computed in $O(\Sort(k))$ IOs, the complete algorithm also still runs in $O(\Sort(|N|+|E|))$ IOs in the worst case.

In many practical settings the priority queues in the first phase may be small and fit in main memory. For example, if the input graph is a tree in reverse depth-first order, then at any time during phase one, the queues will only contain messages to/from the nodes on a single path between the root and a leaf. As long as the children of a single node always fit in memory, the hash values can be computed in memory as well. Thus, the cost of computing the hash values in the first phase is small, and in practice, each hash value will be read or written at most eight times when writing, sorting, and reading $\mathit{Ranks}$ and $N'$. In return, the grouping by rank, label, \emph{and} hash value induces a much finer partitioning of $G$ than the grouping by rank and label only. As a result, the sorting on line 14 of the second phase will be less likely to require the use of external memory. Note that each node's record in the $\mathit{Group}$ file contains one number for the node's original identifier, plus one number for each neighbour of the node (a bisimilarity identifier for every child, and a node identifier for every parent). Therefore, even if on average, nodes have only two neighbours, a record in the group file has an average size of three numbers. Since sorting out-of-memory would take at least two read passes and two write passes, this would amount to the transfer of $3 \cdot 4 = 12$ numbers per node to or from disk. Thus, even in this setting with few edges, the optimization with hashing may already lead to IO-savings in the second phase of twelve numbers per node.

\subsection{Implementation using STXXL}\label{subsec:stxxl}

We have implemented the enhanced bisimulation partitioning algorithm of Section
\ref{subsec:enhancedalgo} using the building blocks available in the STXXL
library, a mature open-source C++ library which provides basic external memory
data structures and algorithms \cite{stxxl}.\footnote{For more information we
refer to the STXXL project page at {\tt http://stxxl.sourceforge.net/}.} Since
STXXL does not include algorithms to sort sets of variable-length records in
external memory, we used the following adaptation of Algorithm~\ref{alg:exbisim}.

\ignore{
Because in practice, the enhanced algorithm ensures
that the sorting on line 14 of Algorithm~\ref{alg:exbisim} will rarely run out
of memory, we did not (yet) actually implement the necessary string sorting
algorithms; instead we used the following adaptation of Algorithm~\ref{alg:exbisim}.
}


Instead of storing, for each node $n$, a record of the form $(\mathit{bisimFamily}, \mathit{origId}, \mathit{parents})$ in the file $\mathit{Group}$, we store the following fixed-size records:
(i) one record of the type $(\mathit{secondHash},$ $\mathit{origId})$, where $\mathit{secondHash}$ is a secondary hash value computed from the bisimilarity family of $n$;
(ii) for each child of $n$, a record of the type $(\mathit{secondHash},$ $\mathit{origId},$ $\mathit{childsBisimId})$ (these records collectively store the bisimilarity family of $n$);
(iii) for each parent of $n$, a record of the type $(\mathit{secondHash},$ $\mathit{origId},$ $\mathit{parentId})$ (these records collectively store the parents of $n$).
The secondary hash values computed from the bisimilarity families are such that bisimilarity families of different size always have different secondary hash values.

In line 14, we sort the above mentioned records lexicographically, thus
obtaining a list of nodes with their bisimilarity families and parents, ordered
by secondary hash value. Although unlikely, collisions may occur: nodes with
the same secondary hash value (which appear consecutively in the sorted list)
may have different bisimilarity families. Therefore we have to be a bit more
careful when assigning bisimilarity identifiers in line 16--20: when processing
nodes that have the same secondary hash value, we record their bisimilarity
families with their bisimilarity identifiers in a dictionary; before assigning a
new identifier to a particular bisimilarity family, we first check the
dictionary to see if an identifier for this bisimilarity family had already been
assigned. Considering that in practice (and, under the assumption of perfect
hashing, also in theory) the dictionary is unlikely to ever be large,
we used a simple sequential file implementation for this dictionary.
Bisimilarity families are only stored in the dictionary as long as nodes with
the same secondary hash value are being processed; the dictionary is always
erased before proceeding to nodes with another rank, label, hash value, or
secondary hash value.

\ignore{
\HH{I wrote a small paragraph, commented out in the file, with the analysis that the algorithm still has a good expected number of IOs---but I think I prefer not to mention that. I am reluctant to bring up the word ``expected'', I do not want the referees to start to have doubts about the other (worst-case) bounds. And the result is not so interesting anyway, since we know that we could do better in theory \emph{and} practice by implementing string-sorting.}

\GF{I completely agree.}

Under the assumption that the hash values computed from the bisimilarity families are such that the expected maximum number of families that have the same hash value is bounded by a constant, the IO-cost of any query in the dictionary (performed by simply scanning the complete dictionary) is always proportional to the size of the query family divided by $B$. Since the total size of all bisimilarity families encountered during a complete run of the algorithm is $|E|$, the queries in the dictionary would take an expected $O(\Scan(|E|))$ IOs in total. Thus, the expected number of IOs needed by the complete bisimulation partitioning algorithm is still $O(\Sort(|N|+|E|))$.
}

\section{Indexing XML documents}
\label{sec:xml}

\begin{figure*}
    \centering
    \subfloat[XML Document tree]{
        \begin{tikzpicture}[graph]
            \node (n0) at (1.75,3) {a$_{(0)}$};
            \node (n1) at (0.5,2) {a$_{(1)}$};
            \node (n2) at (3,2) {a$_{(2)}$};

            \node (n3) at (0,1) {b$_{(3)}$};
            \node (n4) at (1,1) {c$_{(4)}$};
            \node (n5) at (2,1) {b$_{(5)}$};
            \node (n6) at (3,1) {c$_{(6)}$};
            \node (n7) at (4,1) {a$_{(7)}$};

            \node (n8) at (3.5,0) {b$_{(8)}$};
            \node (n9) at (4.5,0) {c$_{(9)}$};

            \path[->] (n0) edge (n1)
                           edge (n2)
                      (n1) edge (n3)
                           edge (n4)
                      (n2) edge (n5)
                           edge (n6)
                           edge (n7)
                      (n7) edge (n8)
                           edge (n9);
        \end{tikzpicture}
    }\qquad
    \subfloat[1-index]{
        \begin{tikzpicture}[graph]
            \node (n0) at (1.25,3) {a$_{(0)}$};
            \node (n1) at (1.25,2) {a$_{(1,2)}$};

            \node (n3) at (0,1) {b$_{(3,5)}$};
            \node (n4) at (1.25,1) {c$_{(4,6)}$};
            \node (n7) at (2.5,1) {a$_{(7)}$};

            \node (n8) at (1.75,0) {b$_{(8)}$};
            \node (n9) at (3.25,0) {c$_{(9)}$};

            \path[->] (n0) edge (n1)
                      (n1) edge (n3)
                           edge (n4)
                           edge (n7)
                      (n7) edge (n8)
                           edge (n9);
        \end{tikzpicture}
    }\qquad
    \subfloat[F\&B-index]{
        \begin{tikzpicture}[graph]
            \node (n0) at (1.25,3) {a$_{(0)}$};
            \node (n1) at (0.75,2) {a$_{(1)}$};
            \node (n2) at (1.75,2) {a$_{(2)}$};

            \node (n3) at (0,1) {b$_{(3,5)}$};
            \node (n4) at (1.25,1) {c$_{(4,6)}$};
            \node (n7) at (2.5,1) {a$_{(7)}$};

            \node (n8) at (1.75,0) {b$_{(8)}$};
            \node (n9) at (3.25,0) {c$_{(9)}$};

            \path[->] (n0) edge (n1)
                           edge (n2)
                      (n1) edge (n3)
                           edge (n4)
                      (n2) edge (n3)
                           edge (n4)
                           edge (n7)
                      (n7) edge (n8)
                           edge (n9);
        \end{tikzpicture}
    }\\
    \subfloat[$A(0)$-index]{
        \begin{tikzpicture}[graph]
            \node (n1) at (1,1) {a$_{(0,1,2,7)}$};

            \node (n3) at (0,0) {b$_{(3,5,8)}$};
            \node (n4) at (2,0) {c$_{(4,6,9)}$};

            \path[->] (n1) edge (n3)
                           edge (n4)
                           edge[loop right] (n1);
            \end{tikzpicture}
    }\qquad
    \subfloat[$A(1)$-index]{
        \begin{tikzpicture}[graph]
            \node (n0) at (1,3) {a$_{(0)}$};
            \node (n1) at (1,2) {a$_{(1,2,7)}$};

            \node (n3) at (0,1) {b$_{(3,5,8)}$};
            \node (n4) at (2,1) {c$_{(4,6,9)}$};

            \path[->] (n0) edge (n1)
                      (n1) edge (n3)
                           edge (n4)
                           edge[loop right] (n1);
        \end{tikzpicture}
    }\qquad
    \subfloat[$A(2)$-index]{
        \begin{tikzpicture}[graph]
            \node (n0) at (0.25,3) {a$_{(0)}$};
            \node (n1) at (0.25,2) {a$_{(1,2)}$};

            \node (n3) at (0,1) {b$_{(3,5,8)}$};
            \node (n4) at (2,1) {c$_{(4,6,9)}$};
            \node (n7) at (1.75,2) {a$_{(7)}$};

            \path[->] (n0) edge (n1)
                      (n1) edge (n3)
                           edge (n4)
                           edge (n7)
                      (n7) edge (n3)
                           edge (n4);
        \end{tikzpicture}
    }\qquad
    \subfloat[$A(3)$-index]{
        \begin{tikzpicture}[graph]
            \node (n0) at (1.25,3) {a$_{(0)}$};
            \node (n1) at (1.25,2) {a$_{(1,2)}$};

            \node (n3) at (0,1) {b$_{(3,5)}$};
            \node (n4) at (1.25,1) {c$_{(4,6)}$};
            \node (n7) at (2.5,1) {a$_{(7)}$};

            \node (n8) at (1.75,0) {b$_{(8)}$};
            \node (n9) at (3.25,0) {c$_{(9)}$};

            \path[->] (n0) edge (n1)
                      (n1) edge (n3)
                           edge (n4)
                           edge (n7)
                      (n7) edge (n8)
                           edge (n9);
        \end{tikzpicture}
    }
    \caption{An XML document tree and five different index types; namely the 1-index, the F\&B-index, and the $A(k)$-index (for $0 \leq k \leq 3$). We have annotated each node in the XML document tree with a unique identifier. This identifier is used in the indices to indicate the nodes represented by each index node.}\label{fig:xml_indices}
\end{figure*}
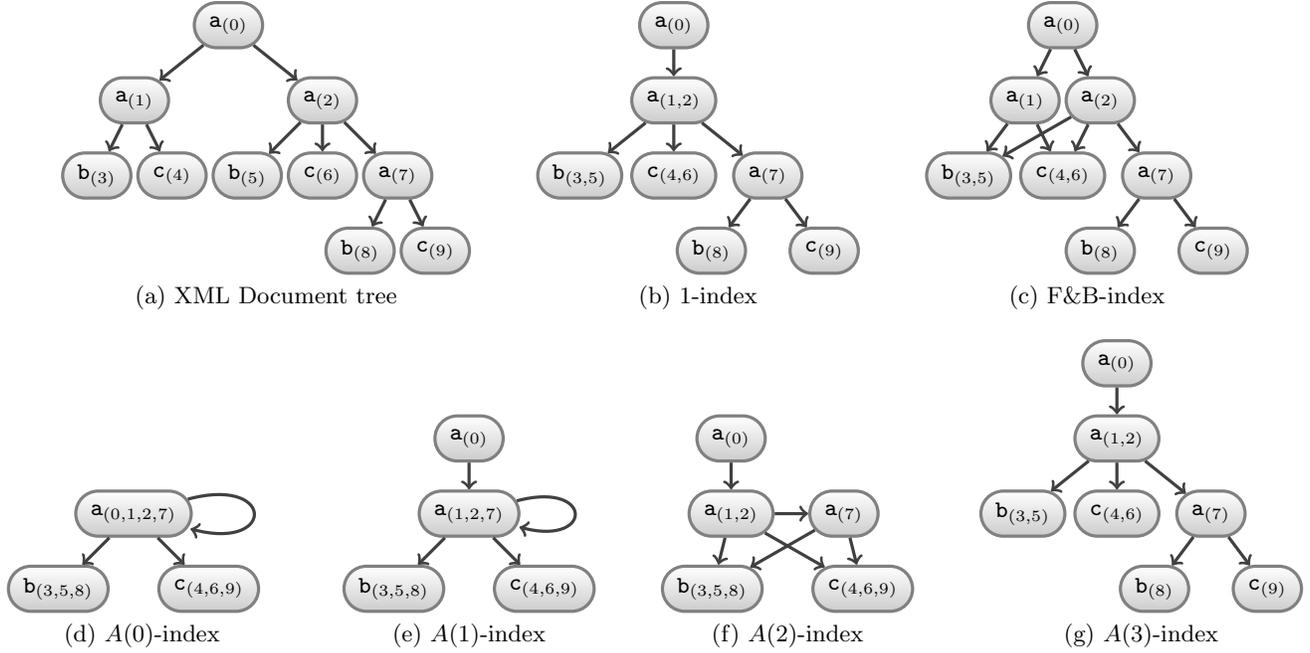

XML documents are widely used to exchange and store tree-structured data
\cite{fbbook}.  In this section we investigate specializations of the general algorithm
from the previous section to efficiently calculate in external memory variations
of bisimulation which have been proposed in the design of indexing data
structures for XML and semi-structured databases.  We shall discuss two
well-known variations, namely the 1-index \cite{indexbb} and the $A(k)$-index
\cite{akindex}.
We also briefly discuss how our approach can be specialized to efficiently
compute the well known F\&B-index \cite{fbbook}.


In Figure \ref{fig:xml_indices} we have given an example of a simple XML
document tree and indices built from this tree.  In the
index figures, nodes represent partition blocks, and there exists an edge
from block $A$ to block $B$ if (and only if) there exists an edge in the
original document from a node in $A$ to a node in $B$.

\subsection{The 1-index}\label{sec:xmlbbisim}

The 1-index utilizes ``backward'' bisimulation to relate nodes with the same
structure with respect to path-queries \cite{indexbb}. Backward bisimulation is equivalent to
normal bisimulation on a graph wherein all edges are reversed in direction.
Figure \ref{fig:xml_indices}(b) illustrates the 1-index of our example XML
document.

Backward bisimulation combined with the nested tree-struc\-ture  of XML documents
gives us several properties that we can utilize to optimize bisimulation
partitioning.
Recall that the basic algorithm from Section~\ref{subsec:basicalgo} consists of
two phases: in the first phase nodes are sorted by rank and label; in the second
phase nodes of the same rank and label are sorted by bisimilarity family.
Alternatively we could take the following approach: in the first phase we sort
by rank only; in the second phase we sort nodes of the same rank by label and
bisimilarity family. Obviously, this would not affect the correctness and the
asymptotic I/O-complexity of the algorithm. However, in the case of 1-indexes
for XML-documents it brings the following advantage: we can avoid the use of a
priority queue in the first phase, and we can avoid use of several sorting
passes to assign identifiers to nodes. We achieve this as follows.

Instead of assigning identifiers to nodes by first sorting the nodes by rank and
then using the positions in the sorted list as identifiers, we can use composite
node identifiers of the form $(\mathit{rank}, \mathit{idOnLevel})$, where $\mathit{rank}$ is
the backward rank of the node (that is, the node's depth in the tree), and
$\mathit{idOnLevel}$ is a unique identifier with respect to all nodes with
backward rank $\mathit{rank}$.  We can now use the structure of XML documents to
compute these identifiers efficiently---in particular we will exploit the fact
that an XML document essentially stores a so-called \emph{Euler
tour}~\cite{meyer} of the tree, in order.

Our algorithm will traverse the tree while maintaining a counter $\mathit{depth}$ that holds the depth of the current position in the tree, and an array $\mathit{count}$, in which the $i$-th number (denoted $\mathit{count}[i]$) holds the number of nodes at depth $i$ encountered so far. Initially, $\mathit{depth} = 0$ and the array $\mathit{count}$ is empty; whenever we try to read an element of $\mathit{count}$ that does not exist yet, this element will be created and initialized to zero.

\begin{algorithm}[htb!]
\caption{XML 1-index, phase 1: sort by depth}\label{alg:brank}
\begin{algorithmic}[1]
\REQUIRE XML document $D$.
\ENSURE  file $N'$ of XML nodes as records\\ $(\mathit{rank},
\mathit{idOnLevel}, \mathit{origId}, \mathit{label})$,\\
          sorted by $(\mathit{rank}, \mathit{idOnLevel})$;\\
          file $E'$ of edges as records\\ $(\mathit{parentRank},
\mathit{parentIdOnLevel}, \mathit{childIdOnLevel})$,\\
          sorted by $(\mathit{parentRank}, \mathit{parentIdOnLevel})$;\\
\SPACE
\STATE create empty file $N'$
\STATE create empty file $E'$
\STATE create empty array of counters $\mathit{count}$
\STATE $\mathit{depth} \gets 0$
\SPACE
\FORALL{tags $\mathit{tag}$ of $D$, in order}
     \IF{$\mathit{tag}$ is a start tag}
         \IF{$\mathit{count}[\mathit{depth}]$ does not exist}
             \STATE add an entry $\mathit{count}[\mathit{depth}] = 0$ to
$\mathit{count}$
         \ENDIF
         \IF{$\mathit{depth} \neq 0$}
             \STATE append\\
             \quad $(\mathit{depth}-1,
\mathit{count}[\mathit{depth}-1]-1, \mathit{count}[\mathit{depth}])$\\
             \quad to $E'$
         \ENDIF
         \STATE determine node identifier $\mathit{origId}$ and label
$\mathit{label}$
         \STATE append $(\mathit{depth}, \mathit{count}[\mathit{depth}],
\mathit{origId}, \mathit{label})$ to $N'$
         \STATE increment $\mathit{count}[\mathit{depth}]$
         \STATE increment $\mathit{depth}$
     \ELSIF{$\mathit{tag}$ is an end tag}
         \STATE decrement $\mathit{depth}$
     \ENDIF
\ENDFOR
\SPACE
\STATE sort $N'$ by $(\mathit{rank}, \mathit{idOnLevel})$
\STATE sort $E'$ by $(\mathit{parentRank}, \mathit{parentIdOnLevel})$
\SPACE
\RETURN $(N', E')$
\end{algorithmic}
\end{algorithm}

Now, when we read a start-tag representing a node $n$ during the processing of an XML document, this node is assigned backward rank $\mathit{rank} = \mathit{depth}$ and $\mathit{idOnLevel} = \mathit{count}[\mathit{rank}]$; 
we increment both $\mathit{count}[\mathit{rank}]$ and $\mathit{depth}$ by one, and 
we construct an edge to $n$ from its parent: this must be the last node encountered on the previous level, with composite identifier $(\mathit{rank}-1, \mathit{count}[\mathit{rank}-1]-1)$. 
When we read an end-tag we simply decrement $\mathit{depth}$ by one.
After reading the complete tree, we simply sort the nodes by composite identifier, and the edges by the composite identifiers of the parents. Pseudocode for the complete first phase of the algorithm is given in Algorithm~\ref{alg:brank}.

The second phase of the algorithm is now simple to implement. Note that we are computing \emph{backward} bisimilarity equivalence classes, and therefore parents and children have switched roles. Thus, the bisimilarity family of a node is simply the bisimilarity identifier of the parent of a node, and no implementations of string sorting or secondary hash functions and dictionaries (as in Section~\ref{subsec:stxxl}) are needed. Pseudocode is given in Algorithm~\ref{alg:xmlbbisim}.


\begin{theorem}\label{thm:1index}
Given an XML-document of $N$ nodes, we can compute its 1-index in $O(\Sort(|N|))$ IOs.
\end{theorem}
\begin{proof}
We use Algorithm~\ref{alg:brank}, followed by Algorithm~\ref{alg:xmlbbisim}. The correctness of the algorithm follows from the same arguments as for Algorithm~\ref{alg:sortbyrank} and Algorithm~\ref{alg:exbisim} in Theorem~\ref{thm:dagfps_exp}.

As for the IO-complexity, observe that the accesses to the file $\mathit{count}$
follow a very well-structured pattern: effectively we move ahead in the file by
one step whenever we encounter a start tag, and we move back in the file by one
step whenever we encounter an end tag. Thus, if we keep the two most recently
accessed blocks of the file in memory, at least $B$ tags must be read between
successive IOs on the $\mathit{count}$ file. Otherwise, the algorithm runs in
$O(\Sort(|N| + |E|)$ IOs by the same arguments as for
Theorem~\ref{thm:dagfps_exp}; since $|E| = |N| - 1$, this simplifies to
$O(\Sort(|N|))$ IOs.
\end{proof}

\begin{algorithm*}
\caption{XML 1-index, phase 2: sort by backward \bsm\ equivalence class}\label{alg:xmlbbisim}
\begin{algorithmic}[1]
\REQUIRE file of XML nodes $N'$ as records $(\mathit{rank},
\mathit{idOnLevel}, \mathit{origId}, \mathit{label})$,
          sorted lexicographically by $(\mathit{rank},
\mathit{idOnLevel})$;\\
          file of edges $E'$ as records $(\mathit{parentRank},
\mathit{parentIdOnLevel}, \mathit{childIdOnLevel})$,\\sorted
lexicographically by $(\mathit{parentRank}, \mathit{parentIdOnLevel})$;\\
\ENSURE  file of XML nodes $B$ as records $(\mathit{origId},
\mathit{bisimId})$
\SPACE
\STATE create empty file $B$ of records $(\mathit{origId}, \mathit{bisimId})$
\STATE create priority queue $Q$ of records $(\mathit{rank},
\mathit{idOnLevel}, \mathit{parentBisimId})$, ordered by
$(\mathit{rank}, \mathit{idOnLevel})$
\STATE insert $(-1, 0, 0)$ in $Q$ (sentinel for root)
\STATE $\mathit{lastBisimId} \gets 0$
\STATE create empty file $\mathit{Group}$ of records $(\mathit{label},
\mathit{parentBisimId}, \mathit{origId}, \mathit{children})$
\SPACE
\FORALL{$(r, \mathit{idOnLevel}, \mathit{origId}, \mathit{label}) \in
N'$, in order}
     \STATE extract $(r, \mathit{idOnLevel}, \mathit{parentBisimId})$
from $Q$
     \STATE create an empty list $\mathit{children}$
     \WHILE{next record of $E'$ has $\mathit{parentRank} = r$ and
$\mathit{parentIdOnLevel} = \mathit{idOnLevel}$}
         \STATE read $(\mathit{parentRank}, \mathit{parentIdOnLevel},
\mathit{childIdOnLevel})$ from $E$
         \STATE append $\mathit{childIdOnLevel}$ to $\mathit{children}$
     \ENDWHILE
     \STATE append $(\mathit{label}, \mathit{parentBisimId},
\mathit{origId}, \mathit{children})$ to $\mathit{Group}$
     \SPACE
     \IF{$N'$ has no more records with $\mathit{rank} = r$}
         \STATE sort $\mathit{Group}$ lexicographically by
$(\mathit{label}, \mathit{parentBisimId})$
         \FORALL{$(\mathit{label}, \mathit{parentBisimId},
\mathit{origId}, \mathit{children}) \in \mathit{Group}$, in order}
             \IF{$\mathit{label}$ and $\mathit{parentBisimId}$ are not
the same as in previous record of $\mathit{Group}$}
                 \STATE $\mathit{lastBisimId} \gets \mathit{lastBisimId}
+ 1$
             \ENDIF
             \STATE append $(\mathit{origId}, \mathit{lastBisimId})$ to $B$
             \FORALL{$\mathit{childIdOnLevel} \in \mathit{children}$}
                 \STATE insert $(r + 1, \mathit{childIdOnLevel},
\mathit{lastBisimId})$ in $Q$
             \ENDFOR
         \ENDFOR
         \STATE erase $\mathit{Group}$
     \ENDIF
\ENDFOR
\SPACE
\RETURN $B$
\end{algorithmic}
\end{algorithm*}

\subsubsection{The F\&B-index}
 
 The 1-index summarizes the structure of graphs by only looking in one
 direction, from parent to child.  The F\&B-index groups nodes based on a
 summary of their structure with respect to both ancestors and descendants
 \cite{fbbook,fbkaushik}.  
Figure \ref{fig:xml_indices}(c) illustrates the F\&B-index of our example XML
document.
 
 For trees, Grimsmo et al.\ have shown that the
 F\&B-index partitioning can be obtained by first computing forward bisimulation
 and then refining the obtained partition by computing backward bisimulation,
 i.e., by applying the algorithm from Section~\ref{sec:partitioning} twice (once
 with edges reversed) \cite{linear}.   It is possible to significantly reduce
 the cost of this computation,  by a straightforward adaptation of the algorithm from
 Section~\ref{sec:xmlbbisim} for the backwards bisimulation
step \cite{thesis}.

\subsection{The \boldmath{$A(k)$}-index}

The $A(k)$-index utilizes backward node $k$-bisimulation,
a localized variant of backward node bisimulation.
The $A(k)$-index groups nodes $n$ based on the structure of
ancestor nodes at most $k$ steps away from $n$.

\begin{definition}\label{def:kbbisim}
Let $G = \langle N, E, l \rangle$ be a graph, $m, n \in N$, and $k\geq
0$.   We say $m$ and $n$ are
backward $k$-bisimilar, denoted  $n \bisim^k m$, if and only if $k=0$ and
$l(n) = l(m)$, or $k > 0$ and:
    \begin{enumerate}
        \item the nodes $n$ and $m$ are backward $(k-1)$-bisimilar, that is, $n \bisim^{k-1} m$;
        \item for each $n' \in \parents(n)$, there is an $m' \in \parents(m)$ with
        $n' \bisim^{k-1} m'$, and
        \item for each $m' \in \parents(m)$, there is an $n' \in \parents(n)$ with
        $n' \bisim^{k-1} m'$.
    \end{enumerate}
\end{definition}
Figures \ref{fig:xml_indices}(d)-(g) illustrate the 
$A(0)$-, 
$A(1)$-, 
$A(2)$-, and
$A(3)$-index, resp., of our example XML document.

The $A(k)$-index seems similar to the 1-index. However, there is a critical difference
between the two. Whereas all backward
bisimilar nodes have the same rank, this does not necessarily
hold for backward $k$-bisimilar nodes. We thus cannot use backward rank to
localize the partitioning computations.  We can, however, express backward node
$k$-bisimilarity on trees in another way; namely, in terms of $k$-traces.

\begin{definition}
Let $G = \langle N, E, l \rangle$ be a tree,  $r \in N$ be the root of
$G$, $n\in N$, and
$L(r, n) = \langle l(r),\dots,l(n)\rangle$
be the sequence of labels of the nodes on the path from $r$ to $n$ in
$E$.
For $k\geq 0$,
the $k$-trace of $L(r, n)$, denoted $T_n^k$, is the sequence
containing the last $k$ elements in $L(r, n)$.
If $k > |L(r, n)|$,  the length of $L(r, n)$, then the $k$-trace is
constructed by prefixing $L(r, n)$ with
$k - |L(r, n)|$ occurrences of some reserved label $\lambda$ not in the range of
$l$.
\end{definition}

The $k$-traces, which are easily represented by fixed-size values, are used
for identifying backward $k$-bisimilar equivalent nodes, as follows.

\begin{proposition}\label{prop:aktracek1}

Let $G = \langle N, E, l \rangle$ be a tree,
$m,n \in N$, and $k\geq 0$.
Then $n \bisim^k m$ if and
only if $T_n^{k+1} = T_m^{k+1}$.

\end{proposition}

While processing an XML document, we can use a stack to store the labels of all
parents of the current node $n$ by pushing the label of a node onto the stack
when we encounter a start-tag and popping the top of the stack when we encounter
an end-tag. By taking the topmost $k+1$ elements we get $T^{k+1}_n$, the $(k+1)$-trace of $n$.
This leads
to a simple $A(k)$-index construction algorithm for XML documents, a
sketch of which is presented in Algorithm \ref{alg:akbisim}. The
algorithm has a worst-case IO-complexity of $O(\Sort(k|N|))$ IOs.

\begin{algorithm}[htb!]
\caption{$A(k)$-index construction for XML documents}\label{alg:akbisim}
\begin{algorithmic}[1]
\REQUIRE XML document $D$.
\ENSURE file of XML nodes $B$ as records $(\mathit{origId}, \mathit{trace})$
\SPACE
\STATE create empty file $B$ of records $(\mathit{origId}, \mathit{trace})$
\STATE create empty stack $S$
\STATE push $k$ dummy labels $\lambda$ onto $S$
\SPACE
\FORALL{tags $\mathit{tag}$ of $D$, in order}
     \IF{$\mathit{tag}$ is a start tag}
        \STATE determine node identifier $\mathit{origId}$ and label
$\mathit{label}$
        \STATE push $\mathit{label}$ onto $S$
        \STATE append ($\mathit{origId}$, top $k+1$ elements of $S$) to $B$
     \ELSIF{$\mathit{tag}$ is an end tag}
        \STATE pop one label from $S$
     \ENDIF
\ENDFOR
\SPACE
\STATE Sort $B$ by $\mathit{trace}$
\RETURN $B$
\end{algorithmic}
\end{algorithm}

\section{Empirical analysis}\label{sec:exp}

To investigate the empirical behavior of our algorithms, we have performed four
separate experiments.
All experiments were performed on a
standard laptop with an Intel Core i5-560M processor and 4GB of main memory. We
have used the internal hard disk drive of this system for sorting and for
storing temporary data structures such as priority queues.  Further details can be
found in \cite{thesis}.\footnote{Open-source code of the full C++ implementation of the
algorithms and supporting tooling used in our analysis can be found at {\tt http://jhellings.nl/projects/exbisim/}.}

\paragraph*{Experiment 1} In this experiment we measured the performance of the
general bisimulation algorithm from Section~\ref{subsec:stxxl} as a function of
the number of nodes in the input graph.  We also considered the difference
between starting with a good initial partition (by
Algorithm~\ref{alg:sortbyhash}, based on rank, label, and hash value) and
starting with a poor initial partition (by Algorithm~\ref{alg:sortbyrank}, based
on rank and label only), in order to measure the impact of a good initial
partition on performance.

For this experiment, random graphs having between $100 \cdot 10^6$ and $1000 \cdot
10^6$ nodes were created using the generator described in Appendix
\ref{sec:generate}.  Every graph had an average of three to four edges per node.
The file size of the input graphs ranged between $2.1$GB and $21.2$GB.  The
number of bisimulation partition blocks in the output ranged from $70 \cdot
10^6$ for the smallest graph to $708 \cdot 10^6$ for the largest graph. For the largest input we have measured a total of 35919 reads from disk; 35091 writes to disk. Thereby a total 70.1GB was read and 68.5GB was written. These measurements only include temporary file usage (priority queues and sorting); not the reading of input and writing of output. In
Figure \ref{fig:ex_size} we have plotted the results for this experiment.

\begin{figure}
    \centering
    \includegraphics[scale=0.9125]{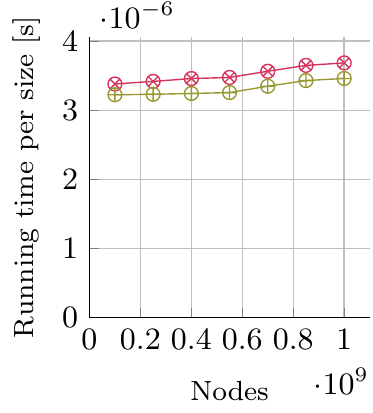}
    \hspace{0.25cm}
    \includegraphics[scale=0.9125]{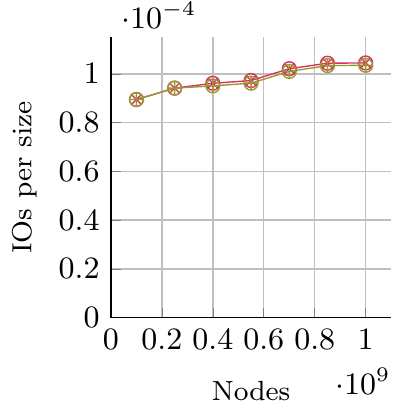}\\
    \includegraphics[scale=0.9125]{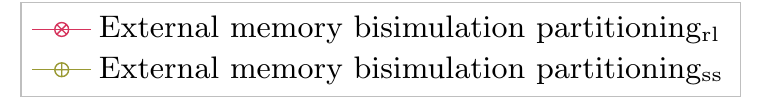}

    \caption{Performance of the bisimulation 
    algorithm from Section~\ref{sec:partitioning} (Experiment 1).
    On the left, running time per node and edge is plotted against the number of
    nodes in the input. On the right, the number of IOs performed per node and
    edge is plotted against the number of nodes in the input. The
    subscript $\mathrm{rl}$ indicates an initial partition based on rank and
    label, the subscript $\mathrm{ss}$ indicates an initial partition based on
    rank, label, and hash value.} \label{fig:ex_size}

\end{figure}

\paragraph*{Experiment 2} In this experiment we measured the performance of the
general bisimulation algorithm from Section~\ref{subsec:stxxl} as a function of
the number of edges in the input graph.  To this end, we created graphs having
$5 \cdot 10^4$ nodes and between $0$ and $1249 \cdot 10^6$ edges, using the
generator as described in Appendix \ref{sec:generate}.  In Figure
\ref{fig:ex_edge} we have plotted the results for this experiment.

\begin{figure}
    \centering
    \includegraphics[scale=0.9125]{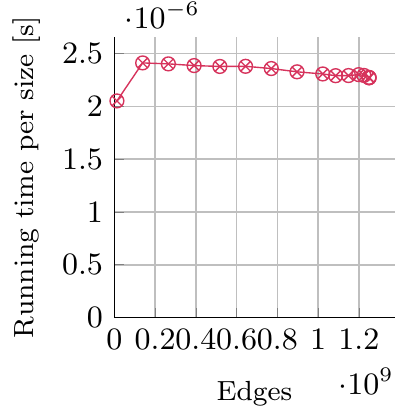}
    \hspace{0.25cm}
    \includegraphics[scale=0.9125]{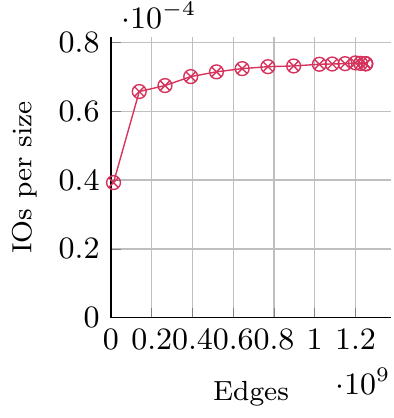}\\
    \includegraphics[scale=0.9125]{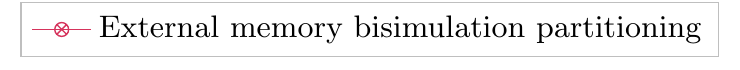}

    \caption{Performance of the bisimulation algorithm from Section~\ref{sec:partitioning} (Experiment 2).
    On the left, running time per node and edge is plotted against the number of
    edges in the input. On the right, the number of IOs performed per node and
    edge is plotted against the number of edges in the input.}
    \label{fig:ex_edge}

\end{figure}

\paragraph*{Experiment 3} In this experiment we measured the performance the
general bisimulation algorithm of Section \ref{subsec:stxxl} on a single graph
as a function of the amount of available memory (per data structure). For this
experiment we fixed a single graph with $10^8$ nodes and $3.3 \cdot 10^8$ edges,
generated as described in Appendix \ref{sec:generate}.  On this graph we
performed external memory bisimulation partitioning, using versions of the
algorithm from Section~\ref{sec:partitioning} constrained to a limited memory
usage. We used values between 12\,MB and 1.5\,GB for the amount of memory the
algorithm is allowed to use. In Figure \ref{fig:ex_mem} we have plotted the
results for this experiment.

\begin{figure}[t]
    \centering
    \includegraphics[scale=0.9125]{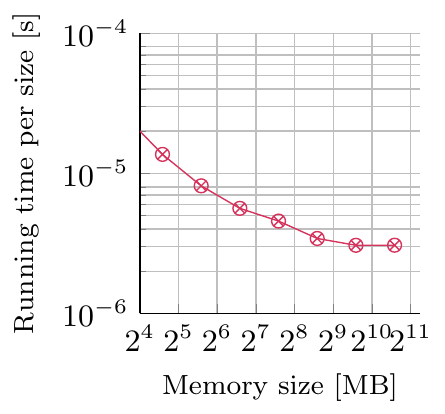}
    \includegraphics[scale=0.9125]{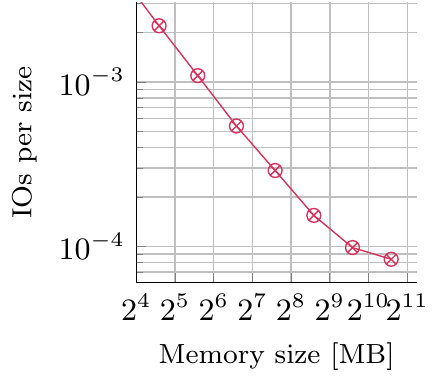}\\
    \includegraphics[scale=0.9125]{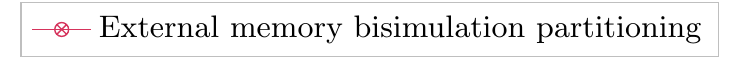}

    \caption{{Impact of available internal memory on the performance of the bisimulation algorithm
    from Section~\ref{sec:partitioning} (Experiment 3).  On the left, running time per node and edge is plotted against
    the amount of available memory. On the right, the number of IOs
    performed per node and edge is plotted against the amount of
    available memory. Note that the amount of available memory excludes stack space used by local variables and the memory used for buffers (256MB in total).}
    }\label{fig:ex_mem}

\end{figure}

\paragraph*{Experiment 4} In this experiment we compared the performance of, on
one hand, the general DAG bisimulation algorithm from
Section~\ref{subsec:stxxl}, and, on the other hand, the specialized algorithm
from Section~\ref{sec:xmlbbisim} for 1-index construction on XML documents. The
performance of both algorithms is measured as a function of the size of the
input graph. For this experiment we created XML documents using the {\tt xmlgen}
program provided by the XML Benchmark Project.\footnote{We have used version
0.92 of {\tt xmlgen}, see {\tt http://www.xml-\\benchmark.org/} for details.}
For the generation of XML documents we have used scaling factors between $50$
and $500$, resulting in documents with sizes between 5.6GB ($10^8$ nodes) and
55.8GB ($10^9$ nodes). In Figure \ref{fig:ex_xml} we have plotted the results
for this experiment.

\begin{figure}
    \centering
    \includegraphics[scale=0.9125]{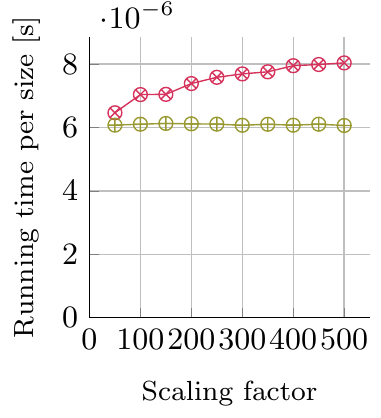}
    \hspace{0.25cm}
    \includegraphics[scale=0.9125]{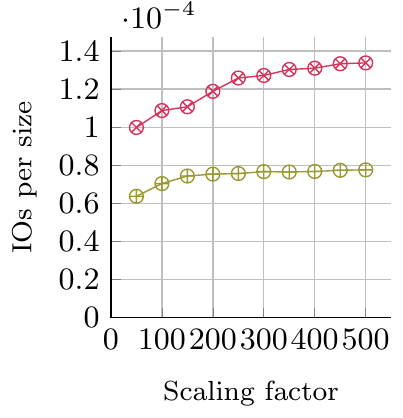}\\
    \includegraphics[scale=0.9125]{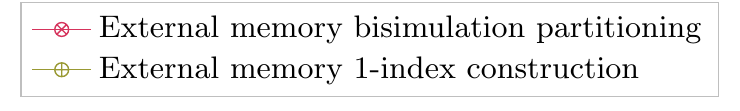}

    \caption{Comparing the performance of the DAG bisimulation algorithm of Section \ref{sec:partitioning} and
    the 1-index algorithm of Section~\ref{sec:xmlbbisim} (Experiment 4).
    On the left, running time per node is plotted as a function of the scaling
    factor. On the right, the number of IOs performed per node is plotted as
    a function of the scaling factor.} \label{fig:ex_xml}

\end{figure}

\paragraph*{Analysis of results} The experiments all show that under all tested
conditions, the general algorithm from Section~\ref{sec:partitioning} is and stays
IO-efficient, even when available memory is artificially limited or when the
number of nodes and edges is very high.  We also see from Experiment 4 that
specializations of our algorithm can outperform the general algorithm with a
good margin.  In particular, we were able to process an 55.8\,GB XML document of
$10^9$ nodes, generated by software from the XMark XML benchmark project, in 104
minutes on a standard laptop with a standard hard disk.

Experiment 1 further shows that a good initial partition (by rank, label and
hash value) improves performance over the less-refined initial partitions (by
rank and label only). A deeper look into the results of this experiment show
that this performance improvement is due to a high increase in the number of
partition blocks in the input for the second phase. This is as expected and does
partly account for the improvement of performance. The results also show a
reduction of the collisions on the secondary hash values that are computed from
bisimilarity families (as explained in Section~\ref{subsec:stxxl}) --- in fact,
collisions were completely eliminated~\cite{thesis}.

Experiment 3 shows that the algorithm does benefit from having more memory at
its disposal.  However, the impact of an increase in available memory becomes
less significant for larger amounts of available memory. 

From the results of the experiments and from the structure of the algorithm we do not expect that certain types of DAGs will have a much better performance than others. An in-depth look into the running time performance shows that it is mainly dominated by the first phase; and within this phase the majority of time is spend on sorting and renumbering the entire graph (last lines of Algorithm \ref{alg:sortbyrank} and Algorithm \ref{alg:sortbyhash}). This sorting and renumbering is unaffected by any particular graph structures.

\section{Concluding Remarks}\label{sec:conclude}

In this paper we have developed the first IO-efficient
bisimulation partitioning algorithm for DAGs.  We also developed specializations
of our general algorithm to compute well-known variants of bisimulation for
disk-resident XML data.   We have complemented our theoretical analysis of these
algorithms with an empirical investigation which established their practicality
on graphs having billions of nodes and edges.

The proposed algorithms are simple enough for practical implementation and use,
for example in the design and
implementation of scalable indexing data structures to facilitate efficient
search  and query answering in a wide variety of real-world applications of
DAG-structured data, as discussed in Section \ref{sec:intro}.

\paragraph*{Future work} The conceptual and practical results developed here
pave the way for a variety of further investigations.  We conclude the paper
with a brief discussion of some of these.

\emph{Generalizing bisimulation partitioning.}  DAGs are adequate for
representing XML data and other practical types of hierarchical data. However,
for some applications, including querying RDF graphs and general graph
databases, cycles in the data are common. Looking at the current state of
general external-memory graph algorithms \cite{meyer}, it is not clear that
solutions for IO-efficient bisimulation partitioning on cyclic graphs are likely
to exist. One can however focus research on heuristic approaches to achieve
acceptable performance in many cases, as is common for many external-memory
algorithms (e.g., \cite{toposort}).   Extending our algorithms with such
heuristics to handle cycles is an interesting direction for further study.

\emph{Partition maintenance.} One can expect that a practical data set might be
subject to modifications over time.  Upon modification, it becomes necessary to
update any bisimulation partition maintained on the data.  Of course, this
maintenance can be performed by throwing out the old partition and computing a
new one from scratch.  It is easy to show that, in the worst-case, partition
maintenance can indeed be as expensive as calculating a new partition from
scratch.  In many practical cases, however, such a drastic approach is
avoidable.  For example, approximations of bisimulation which are cheaper to
maintain may be acceptable.  Internal-memory approaches to incremental partition
maintenance in this spirit have been proposed, e.g.,
\cite{sccbisim,kaushik,incbisim}.  Studying such practical maintenance schemes
for disk-resident data is another interesting direction for future research.

\emph{Practical output formatting.}  In the empirical validation of our
algorithms, we have not considered any particular output format.  An interesting
research problem is to consider adapting our algorithms such that their output
is usefully structured for some intended applications.  For example,  on XML
documents one can explore the combination of our algorithms with the on-disk
data structure studied by Wang et al.\ \cite{fbdisk}.


%
\bibliographystyle{abbrv}
\bibliography{arxiv}  
%
%
\appendix

\section{Generating benchmark data}\label{sec:generate}

Developed as part of our open-source experimental framework \cite{thesis}, the
{\tt gen} program is a benchmark graph generator. The program can be configured
to create random DAGs, trees, chains and transitive-closure chains, with control
of basic features such as node label assignment and graph size.  We used {\tt
gen} to generate the input to Experiments 1--3, discussed in Section
\ref{sec:exp}.  The generator does not try to represent any particular class of
graph structures, instead focusing on the worst-case scenario of random
structure, to stress-test our algorithms.

The program uses a direct approach for generating the input for Experiment 1.
First, {\tt gen} creates $n$ nodes. To each node $v$, {\tt gen} assigns a label
from a limited set of labels that depends on $n$. Then {\tt gen} selects
children to be connected to $v$ by repeatedly flipping a coin that comes up
heads with a certain probability $p$, that is given as a parameter to {\tt
gen}. Whenever the coin comes up heads, a new child for $v$ is selected from the
nodes that were generated before $v$; if the new child was already a child of
$v$, it is ignored. As soon as the coin comes up tails, our program {\tt gen}
stops selecting children for $v$, and moves on to generating the next node.


The program uses a slightly different approach for generating the random graphs
used in Experiment 2. Again, {\tt gen} creates $n$ nodes and assigns labels in
the way described above. To create edges, {\tt gen} considers every pair of
nodes $u,v$ such that $u$ was generated before $v$, and creates an edge from $u$
to $v$ with probability $p$.

For Experiment 3 we use the same approach as for Experiment 1.

\balance
\end{document}